\documentclass{llncs}
\usepackage{llncsdoc}
\usepackage{graphics}

\spnewtheorem{conjecturee}[conjecture]{Conjecture}{\bfseries}{\itshape}

\begin{document}
\markboth{\LaTeXe{} Class for Lecture Notes in Computer
Science}{\LaTeXe{} Class for Lecture Notes in Computer Science}
\thispagestyle{empty}


\pagestyle{plain}

\title{
The Complexity of HCP in Digraps with Degree
Bound Two}
\author{Guohun Zhu}
\institute{  Guilin University of Electronic Technology,\newline
            No.1 Jinji Road,Guilin, Guangxi, 541004,P.R.China \newline
            \email{ccghzhu@guet.edu.cn}}

\maketitle

\begin{abstract}
The Hamiltonian cycle problem (HCP) in digraphs $D$ with degree
bound two is solved by two mappings in this paper. The first
bijection is between an incidence matrix $C_{nm}$ of simple digraph
and an incidence matrix $F$ of balanced bipartite undirected graph
$G$; The second mapping is from  a perfect matching of $G$ to a
cycle of $D$. It proves that the complexity of HCP in $D$ is
polynomial, and finding a second non-isomorphism Hamiltonian cycle
from a given Hamiltonian digraph with degree bound two is also
polynomial. Lastly it deduces $P=NP$ base on the results.
\\
\end{abstract}

\section{Introduction}

It is well known that the Hamiltonian cycle problem(HCP) is one of
the standard NP-complete problem \cite{Johnson1985}. As for
digraphs, even when the digraphs on this case: planar digraphs with
indegree 1 or 2 and outdegree 2 or 1 respectively, it is still
$NP-Complete$ which is proved by J.Plesn{\'\i}k \cite{PLESNIK1978}.


 Let us named a simple strong connected digraphs
with at most indegree 1 or 2 and outdegree 2 or 1 as $\Gamma$
digraphs. This paper solves the HCP of $\Gamma$ digraphs with
following main results.

\begin{theorem}
\label{bipariteofgamma} Given an  incidence matrix $C_{nm}$ of
$\Gamma$ digraph, building a mapping:$F=\left (
    {\begin{array}{c c}
    C^+  \\
    -C^-
    \end{array}}\right)$, then
$F$ is a  incidence
 matrix  of undirected
balanced bipartite graph $G(X,Y;E)$, which obeys the following
properties:
\begin{enumerate}
\item[c1.]
 $|X|=n$,$|Y|=n$,$|E|=m$
\item[c2.]
    $$ \forall x_i \in X \wedge 1 \leq  d(x_i) \leq 2$$
    $$ \forall y_i \in Y \wedge  1 \leq  d(y_i) \leq 2$$
\item[c3.]
   $G$ has at most $\frac{n}{4}$ components which is length of $4$.
\end{enumerate}
\end{theorem}

Let us named the undirected balanced bipartite graph $G(X,Y:E)$ of
$\Gamma$ digraph as projector graph.

\begin{theorem}
\label{perfectofgamma} Let $G$ be the projector graph of a $\Gamma$
graph $D(V,A)$, determining a Hamiltonian cycle in $\Gamma$ digraph
is equivalent to find a perfect match $M$ in $G$ and
$r(C^\prime)=n-1$, where $C^\prime$ is the incidence matrix of
$D^\prime(V,L) \subseteq D$ and $L=\{a_i| a_i \in D \wedge e_i \in
M\}$.
\end{theorem}

Let the each component of $G$ corresponding to a boolean variable, a
 monotonic function $f(M)$ is build to represents the number of
component in $D$. Based on this function, the maximum number of non-isomorphism perfect matching is linear, thus  complexity of $\Gamma$
digraphs has a answer.
\begin{theorem}
\label{complexitygamma} Given the incidence matrix $C_{nm}$ of a
$\Gamma$ digraph , the complexity of finding a Hamiltonian cycle
existing or not is $O(n^4)$
\end{theorem}

The concepts of cycle and rank of graph are given in section $2$.
Then theorems $1$,$2$,$3$ are proved in sections $3$,$4$,$5$
respectively. The last section discusses the $P$ versus $NP$ in more
detail.

\section{Definition and properties}

Throughout this paper we consider the finite simple (un)directed graph
$D=(V,A)$ ($G(V,E)$, respectively), i.e. the graph has no
multi-arcs and no self loops. Let $n$ and $m$ denote the number of
vertices $V$ and arcs $A$ (edges $E$, respectively), respectively.

 As conventional, let $|S|$ denote the number of a set $S$.
 The set of vertices $V$ and set of arcs of $A$ of a digraph $D(V,A)$ are denoted by
 $V=\{v_i | 1 \leq i \leq n\}$ and $A=\{a_j | (1 \leq j \leq m) \wedge a_j=<v_i,v_k>, (v_i \neq v_k \in V) \}$
 respectively,
  where $ <v_i,v_k>$ is a arc from $v_i$ to $v_k$.
  Let the out degree of vertex $v_i$ denoted by $d^{+}(v_i)$,
which has the in degree by denoted as $d^{-}(v_i)$ and has the
degree $d(v_i)$ which equals $d^{+}(v_i)+d^{-}(v_i)$. Let the
$N^+(v_i)=\{v_j| <v_i,v_j> \in A \}$, and $N^-(v_i)=\{v_j |
<v_j,v_i> \in A \}$.

Let us define a forward relation $ \bowtie $ between two arcs as
following, $ a_i \bowtie a_j = v_k \: \mbox{iff}\: a_i=<v_i,v_k>
\wedge a_j=<v_k, v_j> $. It is obvious that $|a_i \bowtie a_i|=0$ .

A {\it cycle } $L$  is a set of arcs $(a_{1},a_{2},\ldots,a_{l})$ in
a digraph $D$, which obeys two conditions:
\begin{enumerate}
\item[c1.] $ \forall a_i \in L,\exists a_j,a_k \in L \setminus \{a_i\},\;  a_i \bowtie a_j \neq a_j \bowtie a_k \in V $
\item[c2.] $ | \bigcup\limits_{a_i \neq a_j \in L } {a_i \bowtie a_j} |=|L|$
\end{enumerate}
   If a cycle $L$ obeys the following conditions, it is a {\it simple cycle}.
\begin{enumerate}
\item[c3.] $\forall L' \subset L $, $L'$ does not satisfy both conditions $c1$  and  $c2$.
\end{enumerate}

 A {\it Hamiltonian cycle $L$} is also a simple cycle of length $n=|V| \geq
 2$ in digraph. As for simplify, this paper given a sufficient condition of Hamiltonian cycle
 in digraph.

\begin{lemma}
\label{HCPdef} If a digraph $D(V,A)$ include a sub graph
$D^\prime(V,L)$ with following two properties, the $D$ is a
Hamiltonian graph.
\begin{enumerate}
\item[c1.] $ \forall v_i \in D^\prime \rightarrow d^+(v_i)=1 \wedge d^-(v_i)=1$,
\item[c2.] $ |L|=|V| \geq 2$ and $D^\prime$ is a strong connected digraph.
\end{enumerate}
\end{lemma}

 A graph that has at least one Hamiltonian cycle is
called a {\it Hamiltonian graph}. A graph G=$(V;E)$ is bipartite if
the vertex set $V$ can be partitioned into two sets $X$ and $Y$ (the
bipartition) such that $\exists e_i \in E, x_j \in X, \forall x_k
\in X \setminus \{x_j\}$, $(e_i \bowtie x_j \neq \emptyset
\rightarrow e_i \bowtie x_k =\emptyset)$ ($e_i, Y$, respectively).
if $|X|=|Y|$, We call that $G$ is a balanced bipartite graph. A
matching $M \subseteq E$ is a collection of edges such that every
vertex of $V$ is incident to at most one edge of $M$, a matching of
balanced bipartite graph is perfect if $|M| = |X|$. Hopcroft and
Karp shows that constructs a perfect matching of bipartite in $O((m
+ n)\sqrt{n})$ \cite{Hopcroft1973}. The matching of bipartite has a
relation with neighborhood of $X$.
\begin{theorem}
\cite{Hall1935}
\label{HallTheorem}
A bipartite graph $G=(X,Y;E)$ has a matching from $X$ into $Y$ if and only if $|N(S)| \geq S$, for any $S \subseteq X$.
\end{theorem}

\begin{lemma}
\label{simplecycle} A even length of simple cycle consist of two
disjoin perfect matching.
\end{lemma}

Two matrices representation for graphs are defined as follows.

\begin{definition}
\cite{Pearl1973} The incidence matrix $C$  of undirected graph $G$
is a two dimensional $n \times m$ table, each row represents one
vertex, each column represents one edge, the $c_{ij}$  in $C$   are
given by
\begin{equation}
c_{ij} = \left\{\begin{array}{ll}
                 1, & \mbox{if $v_i \in e_j$;} \\
                 0,      & \mbox{otherwise.}
                \end{array} \right.
\end{equation}
\end{definition}

It is obvious that every column of an incidence matrix has exactly
two $1$ entries.

\begin{definition}
\label{incidencematrixdef}
\cite{Pearl1973}
 The incidence
matrix $C$  of directed graph $D$ is a two dimensional $n \times m$
table, each row represents one vertex, each column represents one
arc the $c_{ij}$  in $C$   are given by
\begin{equation}
\label{incidencedef}
c_{ij} = \left\{\begin{array}{ll}
                 1, & \mbox{if $<v_i,v_i> \bowtie a_j =v_i $;} \\
                 -1, & \mbox{if $a_j \bowtie <v_i,v_i>=v_i $;} \\
                 0, & \mbox{$otherwise $}.
                \end{array} \right.
\end{equation}
\end{definition}

It is obvious to obtain a corollary of  the incidence matrix as following.

\begin{corollary}
\label{incidencecorollary} Each column of an incidence matrix of
digraph has exactly one $1$ and one $-1$ entries.
\end{corollary}

\begin{theorem}
\label{rank_theorem}
\cite{Pearl1973} The $C$ is the incidence
matrix of a directed graph with $k$ components the rank of $C$ is
given by
\begin{equation}
\label{rankcomponent}
    r(C)=n-k
\end{equation}
\end{theorem}

In order to convince to describe the graph $D$ properties, in this
paper, we denotes the $r(D)=r(C)$.

\section{Divided incidence matrix and Projector incidence matrix}

Firstly, let us divided the matrix of $C$ into two groups.
\begin{equation}
\label{C_Plusedef}
      C^+=\left\{c_{ij} | c_{ij} \geq 0 \mbox{ otherwise  $0$ }\right \}
\end{equation}
\begin{equation}
\label{C_Minusdef}
      C^-=\left\{c_{ij} | c_{ij} \leq 0  \mbox{ otherwise  $0$ } \right \}
\end{equation}

It is obvious that the matrix of $C^+$ represents the forward arc of
a digraph  and  $C^-$ matrix represents the backward arc
respectively.  A corollary is deduced as following.
\begin{corollary}
\label{rank1}
  A digraph $D=(V,A)$ is strong connected if and only if the rank of divided incidence matrix
  satisfies $r(C^+)=r(C^-)=|V|$.
\end{corollary}

Secondly, let us combined the the $C^+$ and $C^-$ as following
matrix.
\begin{equation}
F=\left (
    {\begin{array}{c c}
    C^+  \\
    -C^-
    \end{array}}\right)
\end{equation}

 In more additional, let $F$ represents as an incidence
matrix of undirected graph $G( X,Y;E)$. The $F$ is named as {\it
projector incidence matrix } of $C $ and $G$ is named as {\it
projector graph },  where $X$ represents the vertices $V^+$ of $D$,
$Y$ represents the vertices of $V^-$ respectively. In another words
we build a mapping $F: D \rightarrow G$ and denotes it as $G=F(D)$.
So the $F(D)$ has $2n$ vertices and $m$ edges if $D$ has $n$
vertices and $m$ arcs. We also build up a reverse mapping: $F^{-1}:
G \rightarrow D$ When $G$ is a projector graph. To simplify, we also
denotes the arcs $a_i=F^{-1}(e_i)$, $v_i^+=F^{-1}(x_i)$ and
$v_i^-=F^{-1}(y_i)$.

\subsection{Proof of Theorem~\ref{bipariteofgamma} }
Firstly, let us
prove the theorem~\ref{bipariteofgamma}.

\begin{proof}

\begin{enumerate}
\item[c1.]
 Since $\Gamma$ digraph  is strong connected, then each vertices of $\Gamma$ digraph  has at
 least one
 forward arcs, each row of $C^+$ has at least one $1$ entries, and
 the $U$ represents the $C^+$ , so  $$|U|=n$$
 the same principle of $C^-$, each row of $C^-$ has at least one $-1$ entries, and
 the $V$ represents the $C^-$ , so  $$|V|=n$$
 Since the columns of $F$ equal to the columns of $C$,
 $$|E|=m$$
\item[c2.]
    Since the degree of each $v_i$ of $\Gamma$ digraph is
    $ 1 \leq d^{+}(v_i) \leq 2$,
    $$ \forall u_i \in U \wedge 1 \leq  d(u_i) \leq 2$$
    Since the degree of each $v_i$ of $\Gamma$ digraph is
    $ 1 \leq d^{-}(v_i) \leq 2$,
    $$ \forall v_i \in V \wedge  1 \leq  d(v_i) \leq 2$$
\item[c3.]
    Let us prove by contradiction, suppose there are $k>\frac{n}{4}$
    components with length of $4$ in $G$.
    Since $D$ is strong connected, according to the
    corollary~\ref{rank1}, $r(F)=\frac{3n}{2}-q \geq r(C^+)=n$,
    where $q \geq k$ is number of components (including $k$
    components with length of $4$). Thus $q \leq \frac{n}{2}$, then
    there are only $x$ components without length $4$, where $x$ is
    \begin{equation}
    \label{x_eq}
     x=q-k<\frac{n}{4}
    \end{equation}
    Suppose the remind $x$ components with length of $t$ (at least
    $t$ vertices connected by some edges), then
    $4k+xt=\frac{3n}{2}$. So $tx=\frac{3n}{2}-4k<\frac{n}{2}$.
    According to the equation~\ref{x_eq}, the $t<2$. It is
    contradict that the $D$ is strong connected.

\end{enumerate}
\end{proof}

\subsection{The cycle in digraph corresponding matching in projector graph}
Secondly, let us given the properties after mapping Hamiltonian
cycle $L$ of $D$ into the sub graph $M$ of projector graph $G$.

\begin{lemma}
\label{projector1} If a Hamiltonian cycle $L$ of $D$ mapping into a
forest $M$ of projector graph $G$, the forest $M$ consist of $|L|$
number of trees which has only two node and one edge, and $M$ has a
unique perfect matching.
\end{lemma}

\begin{proof}
Let the $\Gamma$ digraph $D(V,A)$ has a sub digraph $D^\prime(V,L)$
which exists one Hamiltonian cycle and $|L|=n$, the incidence matrix
$C$ of $L$ could be permutation as follows.
  \begin{equation}
          C = \left(
        {\begin{array}{llllll}
            1&0&0&\ldots&0&-1\\
            -1&1&0&\ldots&0&0\\
            0&-1&1&\ldots&0&0\\
            0&0&-1&\ldots&0&0\\
            0&0&0&\ldots&0&0\\
            0&0&0&\ldots&-1&1
        \end{array} } \right ).
\end{equation}
Let $$ F=\left (
    {\begin{array}{c c}
    C^+  \\
    -C^-
    \end{array}}\right)
$$

It is obvious that each row of $F$ has only one $1$ entry and each
column of $F$ has two $1$ entries.

According to theorem~\ref{bipariteofgamma}, $F$ represents a
balanced bipartite graph $G(X,Y;E)$ that each vertex has one edge
connected, and each edge $e_i$ connect on vertex $x_i \in X$ ,
another in $Y$, in another words, $ \exists e_i \in E \, x_j \in X
$,$ \forall x_k \in X \setminus \{x_j\}$, $e_i \bowtie x_j \neq
\emptyset \rightarrow e_i \bowtie x_k =\emptyset
$($e_i,Y$,respectively). According the matching definition, $M$ is a
matching, since $|E|=|L|$, $E$ is a perfect matching. and pair of
vertices between $X$ and $Y$ only has one edge, so $M$ is a forest,
and each tree has only two node with one edge.

\end{proof}

\section{Proof of Theorem~\ref{perfectofgamma} }

\begin{proof}
$\Rightarrow$ Let the $\Gamma$ digraph $D(V,A)$ has a sub digraph
$D^\prime(V,L)$ which is a Hamiltonian cycle and $|L|=n$, let matrix
$C^\prime$ represents the incidence matrix of $D^\prime$, so
$r(C^\prime)=n-1$; According to lemma~\ref{projector1}, the
projector graph $F(D^\prime)$ has a perfect matching, thus $F(D)$
also has a perfect matching.

$\Leftarrow$ Let $G(X,Y;E)$ be a projector graph of the $\Gamma$
graph $D(V,A)$,$M$ is a perfect matching in $G$.  Let
$D^\prime(V,L)$ be a sub graph of $D(V,A)$ and $L=\{a_i |a_i \in D
\wedge e_i \in M\}$.
 Since $r(L)=n-1$,  $D^\prime(V,L)$ is a
strong connected digraph. it deduces that $\forall v_i \in D^\prime
$,$d^+(v_i) \geq 1 \wedge d^-(v_i) \geq 1$. Suppose $\exists v_i \in
D^\prime$, $d^+(v_i) > 1$ ($d^-(v_i) > 1$ respectively), Since
$|M|=n$, it deduces that $\sum_{i=1}^{n}d(v_i)>2n+1$, which imply
that $|L|>n$. this is contradiction with $L=\{a_i |a_i \in D \wedge
e_i \in M\}$ and $|M|=n$. So $\forall v_i \in D^\prime$,
$d^+(v_i)=d^-(v_i)=1$, According the lemma~\ref{HCPdef}, $D^\prime$
has a Hamiltonian cycle.
\end{proof}

\section{
Number of perfect matching in projector graph }

Let us considering the number of perfect matching in $G$ . Firstly,
let us considering a example as shown in figure 1.

\setlength{\unitlength}{1cm}
\begin{picture}(10,4)

\label{org_figure}

\linethickness{0.065mm}
 \put(1, 0){Figure 1. Original Digraph $D$}
%
 \put(1, 1){\vector(0, 1){.8}}
 \put(1, 1){\circle{.5}}
 \put(1.2, 1.4){$a_8$}

 \put(1, 2){\vector(1, 0){.8}}
 \put(1, 2){\circle{.5}}
 \put(1.4, 2.1){$a_{1}$}

 \put(2, 1){\vector(-1,0){.8}}
 \put(2, 1){\circle{.5}}
 \put(1.4, .7){$a_{22}$}

 \put(2, 2){\vector(0, -1){.8}}
 \put(2, 2){\vector(1, 0){.8}}
 \put(2, 2){\circle{.5}}
 \put(2.4, 2.1){$a_2$}
 \put(2.2, 1.4){$a_9$}

 \put(3, 1){\vector(-1,0){.8}}
 \put(3, 1){\vector(0,1){.8}}
 \put(3, 1){\circle{.5}}
 \put(2.4, .7){$a_{21}$}
 \put(3.1, 1.4){$a_{10}$}

 \put(3, 2){\circle{.5}}
 \put(3, 2){\vector(1, 0){.8}}
 \put(3.4, 2.1){$a_{3}$}

 \put(4, 1){\vector(-1,0){.8}}
 \put(4, 1){\circle{.5}}
 \put(3.4, .7){$a_{20}$}

 \put(4, 2){\vector(0, -1){.8}}
 \put(4, 2){\vector(1, 0){.8}}
 \put(4, 2){\circle{.5}}
 \put(4.4, 2.1){$a_4$}
 \put(4.1, 1.4){$a_{11}$}

 \put(5, 1){\vector(-1,0){.8}}
 \put(5, 1){\vector(0,1){.8}}
 \put(5, 1){\circle{.5}}
 \put(4.4, .7){$a_{12}$}
 \put(5.1, 1.4){$a_{19}$}

 \put(5, 2){\circle{.5}}
 \put(5, 2){\vector(1, 0){.8}}
 \put(5.4, 2.1){$a_{5}$}

 \put(6, 1){\vector(-1,0){.8}}
 \put(6, 1){\circle{.5}}
 \put(5.4, .7){$a_{18}$}

 \put(6, 2){\vector(0, -1){.8}}
 \put(6, 2){\vector(1, 0){.8}}
 \put(6, 2){\circle{.5}}
 \put(6.4, 2.1){$a_6$}
 \put(6.1, 1.4){$a_{13}$}

 \put(7, 1){\vector(-1,0){.8}}
 \put(7, 1){\vector(0,1){.8}}
 \put(7, 1){\circle{.5}}
 \put(6.4, .7){$a_{17}$}
 \put(7.2, 1.4){$a_{14}$}

 \put(7, 2){\circle{.5}}
 \put(7, 2){\vector(1, 0){.8}}
 \put(7.4, 2.1){$a_{7}$}

 \put(8, 1){\vector(-1,0){.8}}
 \put(8, 1){\circle{.5}}
 \put(7.4, .7){$a_{16}$}

 \put(8, 2){\circle{.5}}
 \put(8, 2){\vector(0,-1){.8}}
 \put(8.4, 1.4){$a_{15}$}


 \end{picture}

Then the projector graph is shown in figure 2.

\setlength{\unitlength}{1cm}
\begin{picture}(12,4)

\label{new_figure}

\linethickness{0.065mm}
 \put(1, 0){Figure 2. Projector graph $G$}
%
 \put(1, 1){\line(0, 1){1.8}}
 \put(1, 1){\circle{.5}}
 \put(1, 3){\circle{.5}}
 \put(1.1, 1.4){$e_1$}

 \put(1.5, 1){\line(0, 1){1.8}}
 \put(1.5, 1){\circle{.5}}
 \put(1.5, 3){\circle{.5}}
 \put(1.6, 1.4){$e_8$}

 \put(2, 1){\line(0, 1){1.8}}
 \put(2, 1){\circle{.5}}
 \put(2, 3){\circle{.5}}
 \put(2.1, 1.4){$e_{22}$}

 \put(3, 1){\line(0, 1){1.8}}
 \put(3, 1){\line(1, 2){1}}
 \put(3, 1){\circle{.5}}
 \put(3, 3){\circle{.5}}
 \put(2.5, 2.0){$e_{9}$}
 \put(3.2, 2.6){$e_{2}$}

 \put(3.3, .4){$G_1$}

 \put(4, 1){\line(0, 1){1.8}}
 \put(4, 1){\line(-1, 2){1}}
 \put(4, 1){\circle{.5}}
 \put(4, 3){\circle{.5}}
 \put(4.1, 1.4){$e_{10}$}
 \put(3.6, 2.1){$e_{21}$}

 \put(4.6, 1){\line(0, 1){1.8}}
 \put(4.6, 1){\circle{.5}}
 \put(4.6, 3){\circle{.5}}
 \put(4.7, 1.4){$e_3$}

 \put(5.3, 1){\line(0, 1){1.8}}
 \put(5.3, 1){\circle{.5}}
 \put(5.3, 3){\circle{.5}}
 \put(5.4, 1.4){$e_{20}$}

 \put(6, 1){\line(0, 1){1.8}}
 \put(6, 1){\line(1, 2){1}}
 \put(6, 1){\circle{.5}}
 \put(6, 3){\circle{.5}}
 \put(5.8, 2.0){$e_{11}$}
 \put(6.2, 2.6){$e_{4}$}

 \put(6.3, .4){$G_2$}

 \put(7, 1){\line(0, 1){1.8}}
 \put(7, 1){\line(-1, 2){1}}
 \put(7, 1){\circle{.5}}
 \put(7, 3){\circle{.5}}
 \put(7.1, 1.4){$e_{19}$}
 \put(6.6, 2.1){$e_{12}$}

 \put(7.7, 1){\line(0, 1){1.8}}
 \put(7.7, 1){\circle{.5}}
 \put(7.7, 3){\circle{.5}}
 \put(7.7, 1.4){$e_5$}

 \put(8.3, 1){\line(0, 1){1.8}}
 \put(8.3, 1){\circle{.5}}
 \put(8.3, 3){\circle{.5}}
 \put(8.3, 1.4){$e_{18}$}

 \put(9, 1){\line(0, 1){1.8}}
 \put(9, 1){\line(1, 2){1}}
 \put(9, 1){\circle{.5}}
 \put(9, 3){\circle{.5}}
 \put(8.6, 2.0){$e_{6}$}
 \put(9.2, 2.6){$e_{13}$}

 \put(9.3, .4){$G_3$}

 \put(10, 1){\line(0, 1){1.8}}
 \put(10, 1){\line(-1, 2){1}}
 \put(10, 1){\circle{.5}}
 \put(10, 3){\circle{.5}}
 \put(9.6, 2.1){$e_{17}$}
 \put(10.3, 1.4){$e_{14}$}
 \put(10.6, 2.0){$\ldots$}

 \end{picture}

Given a perfect matching $M$, each component(cycle) in $G$ has two
partition edges belong to $M$. Let us code component $G_i$ which
$|G_i|>2$ and matching $M$ to a binary variable.
\begin{equation}
G_{i} = \left\{\begin{array}{ll}
                 1, & \mbox{if $G_i \cap M =\{e_j,e_k,\ldots \} $;} \\
                 0, & \mbox{if $G_i \cap M =\{e_l,e_q,\ldots \} $.}
                \end{array} \right.
\end{equation}

Now there are two cases for the number of perfect matching.
\begin{enumerate}
\item [Label edge.] In that cases, the $Code(M_1)=\{0,0,1\}$ is different
with $Code(M_2)=\{0,1,0\}$. If there are $k$ number of
components(cycles), then there are $2^k$ perfect matching.
\item [Unlabel edge.] In that cases, the $Code(M_1)=\{0,0,1\}$ is
isomorphic to $Code(M_2)=\{0,1,0\}$. The same principle that
$Code(M_3)=\{0,1,1\}$ is isomorphic to $Code(M_4)=\{1,1,0\}$ but is
not isomorphic to $Code(M_1)$.
\end{enumerate}

Then let us summary the maximal number of perfect matching in these
two cases.

\begin{lemma}
\label{number_of_matching} The maximal number of labeled  perfect
matching in a projector graph $G$ is $2^{\frac{n}{4}}$, but the
maximal number of unlabeled perfect matching in a projector graph
$G$ is $\frac{n}{2}$.
\end{lemma}

\begin{proof}
According to the theorem~\ref{bipariteofgamma}, there at most
$\frac{n}{4}$ components with a components which is length of $k=4$.
When $k$=2, there are only one perfect matching in $G$; When $k=4$,
there are $\frac{n}{4}$ components which is $C_4$, and so on when
$k=6$, there are $\frac{n}{6}$ components which is $C_6$, etc, so
on. According to the lemma~\ref{simplecycle}, each simple cycle has
divided the perfect matching into two class. So maximal number
perfect matching in the non isomorphism cycle which is
$2^{\frac{n}{4}}$. Since in unlabeled cases, every $C_4$ cycle is
isomorphism, the maximal number of perfect matching is
$2*\frac{n}{4}=\frac{n}{2}$.
\end{proof}

Review the example 1 again, it is easy find that follow proposition.
\begin{proposition}
\label{myprop}Given two perfect matching $M1$ and $M2$ in projector
graph $G$, if
 $code(M1)=code(M2)$, then the $r(F^{-1}(M1))=r(F^{-1}(M2))$.
\end{proposition}

\subsection{Proof of Theorem~\ref{complexitygamma} } Now let us proof
the theorem~\ref{complexitygamma}.
\begin{proof}
Let $G$ be a project balanced bipartition of $D$. According
theorem~\ref{bipariteofgamma}, the $\Gamma$ graph is equivalent to
find a perfect match $M$ in a project $G$.

According to the lemma~\ref{number_of_matching}, the maximal number
non isomorphism perfect matching in $G$ is only $n$.

Thus it is only need exactly enumerate all of non isomorphism
perfect matching $M$, then obtain the $value=r(F^{-1}(M))$,if
$value=n-1$, then the $e_i \in M$ is also $e_i \in C$, where $C
\subset D$ is a Hamiltonian cycle.

Since the complexity of rank of matrix is $O(n^3)$, finding a simple
cycle in a component with degree $2$ is $O(n^2)$, and obtaining a
perfect matching of a bipartite graph is $O((m+n)\sqrt{n})<O(n^2)$
\cite{Hopcroft1973}. Then all exactly algorithms need to calculate
the $n$ time $o(n^3)$. Thus the complexity is $O(n^4)$.
\end{proof}

Since the non isomorphism perfect matching comes from the coding of
edges in the component of $G$, it is not easy implementation.

Let us give two recursive equation to obtain a
 perfect matching $M$ from $G$. Suppose there are $k$ component $G_1,G_2,\ldots G_k$ in
  $G$ where $G_i$ is a component with degree $2$ and $|E_i| \geq 3$.

\begin{equation}
\label{M0_equation}
  M^\prime= \left\{\begin{array}{ll}
                 M(t) \otimes G_t, & \mbox{ $G_t$ is a cycle }; \\
                 M(t),      & \mbox{otherwise.}
                \end{array} \right.
\end{equation}
\begin{equation}
\label{M1_equation}
  M(t+1)= \left\{\begin{array}{ll}
                 M^\prime, & \mbox{if  $r(F^{-1}(M^\prime)) > r(F^{-1}(M(t)))$
                  }; \\
                 M(t),      & \mbox{otherwise.}
                \end{array} \right.
\end{equation}
where $t \leq k-1$, when $t=0$, $M(0)$ is the initial perfect
matching from $G$.

 When $r(F^{-1}(M(t)))=n-1$, According the
theorem~\ref{bipariteofgamma}, the  $A=F^{-1}(M(t))$ is a
Hamiltonian cycle solution. If all of $r(F^{-1}(M(t)))<n-1$, then
there has no Hamiltonian cycle in $D$.

Since the non isomorphism perfect matching $M$ in $G$ is poset, the
function $r(F^{-1}(M))$ in $G$ is monotonic, so this approach is
exactly approach.

Let us give a example to illustrate the approach in detail.

\begin{example}
Considering the digraph $D$ in figure 1, then the projector graph
$G$ in figure 2.

Let
$M(0)=\{e_1,e_8,e_{22},e_9,e_{10},e_3,e_{20},e_{11},e_{19},e_5,e_{18},e_6,
e_{17},e_7,e_{15},e_{16} \} $.

Thus the $r(F^{-1}(M(0))=n-3$. Let $M^\prime=r(F^{-1}(M(0) \otimes
G_3)$,then $r(F^{-1}(M^\prime)=n-4$, thus $M(1)=M(0)$ and then turn
to $G_2$,$G_1$. At last it obtain the solution.
\end{example}

Considering the equation~\ref{M1_equation}, let it substituted by
following equations when $r(M^\prime)=n-1$ and $t < k-1$.

\begin{equation}
\label{M2_equation} M(t+1)=M^\prime \mbox{   if $r(F^{-1}(M^\prime))
\geq r(F^{-1}(M(t)))$ }
\end{equation}

It is obvious that all non-isomorphism Hamiltonian cycle could
obtain by the  repeat check the equation~\ref{M2_equation} and the
equation $r(M^\prime)=n-1$.

In conversely, if a Hamiltonian cycle of $\Gamma$ digraphs is given,
it represents a perfect matching $M$ in its projector graph $G$.
Thus the equation~\ref{M2_equation} and
Theorem~\ref{complexitygamma} follows a corollary.

\begin{corollary}
\label{complexityanotherHCP} Given a Hamiltonian $\Gamma$ digraph,
the complexity of determining another non-isomorphism Hamiltonian
cycle is polynomial time.
\end{corollary}

\subsection{ The HCP in digraph with bound two}

Let us extend the Theorem~\ref{complexitygamma}  to  digraphs with
$d^+(v) \leq 2$ and $d^-(v) \leq 2$ in this section.

\begin{theorem}
\label{complexityofHCP} The complexity of finding a Hamiltonian
cycle existing or not in digraphs with degree $d^+(v)\leq 2$ and
$d^-(v)\leq 2$ is polynomial time.
\end{theorem}

\begin{proof}

 Suppose a digraph $D(V,A)$ having a vertex $v_i$ is shown as figure $3$, which is $d^(v_i)
=2  \wedge d^-(v_i) =2 $

\setlength{\unitlength}{1cm}

\begin{picture}(8,4)

\label{figure2}

\linethickness{0.065mm}

 \put(1, 0){Figure 3. A vertex with degree than 2}
%
 \put(1, 1){\vector(2, 1){1.8}}


 \put(1, 3){\vector(2,-1){1.8}}

 \put(3, 2){\circle{.5}}

 \put(1.2, 1.4){$a1$}
 \put(1.2, 3.1){$a2$}

 \put(3.2, 2){\vector(2, 1){1.8}}
 \put(3.2, 2){\vector(2,-1){1.8}}

 \put(4.2, 1.0){$a3$}

 \put(4.2, 2.8){$a4$}


 \end{picture}

 Let us spilt this vertex to  two vertices that one of vertex has degree with in degree 2 or out degree 1 ,
another vertex has degree with in degree 1 or out degree 2 as shown
in figure $4$. Then the $D$ is derived to a new $\Gamma$ graph $S$.

\begin{picture}( 10,4)
  \label{figure2}
\linethickness{0.065mm}

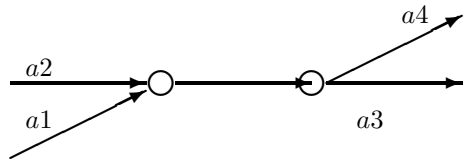
\put(1, 0){Figrue 4 A vertex in $D$ is mapping to a vertex in
$\Gamma$ digraph}
\thicklines

\put(1, 1){\vector(2, 1){1.8}}

\put(1, 2){\vector(1, 0){1.8}}

\put(3, 2){\circle{.3}}

\put(1.2, 1.4){$a1$} \put(1.2, 2.1){$a2$}

\put(3.2, 2){\vector(1,0){1.8}}

\put(5.6, 1.4){$a3$}

\put(5,2){\circle{.3}} \put(6.2, 2.8){$a4$}
\put(5.2,2){\vector(2,1){1.8}}
\put(5.2, 2){\vector(1,0){1.8}}

 %
\end{picture}

It is obvious that each vertex in the $\Gamma$ graph $S$ has
increase $1$ vertices and $1$ arcs of $D$. Suppose the worst cases
is each vertex in $D$ has in degree 2 and out degree 2, the total
vertices in $S$ has $2n$ vertices.

According to the theorem~\ref{complexitygamma},  obtain a Hailtonian
cycle $L^\prime$ in $S$ is no more then $O(n^4)$, then the $D$ will
has a Hamiltonian cycle $L^\prime=L \cap A$.

\end{proof}

\section{
Discussion P versus NP}

The $P$ versus $NP$ is a famous open problem in computer science and
mathematics, which means to determine whether very language accepted
by some nondeterministic algorithm in polynomial time is also
accepted by some deterministic algorithm in polynomial time
\cite{cook2000}. Cook give a proposition for the  $P$ versus $NP$.

\begin{proposition}
\label{PeqNP} If L is NP-complete and $L \in P$, then $P=NP$.
\end{proposition}

According above proposition and the result above section, $P$ versus
$NP$ problem has a answer.

\begin{theorem}
\label{PeqNP_theorem} $P=NP$
\end{theorem}

\begin{proof}
As the result of \cite{PLESNIK1978}, the complexity of HCP in
digraph with bound two is $NP-complete$. According the
theorem~\ref{complexityofHCP}, the complexity of HCP in digraph with
bound two is also $P$, thus according to proposition~\ref{PeqNP},
$P=NP$.

\end{proof}

In fact, the \cite{PLESNIK1978} proves that $ 3SAT \preceq _{p} HCP
\; of \; \Gamma \; digraph$, since $3SAT$ is a $NPC$ problem,  which
also implies that $P=NP$.

\section{Conclusion}
According to the theorem~\ref{complexityofHCP}, the complexity of
determining a Hamiltonian cycle existence or not in digraph with
bound degree two is in polynomial time. And according to the
theorem~\ref{PeqNP_theorem}, $P$ versus $NP$ problem has closed,
$P=NP$.

\section*{Acknowledgements}
The author would like to thank Prof. Kaoru Hirota for valuable
suggestions, thank Prof. J{\o}rgen Bang-Jensen who called mine
attention to the paper \cite{PLESNIK1978}, and thank Andrea Moro for
useful discussions.


\end{document}